\documentclass[a4paper,web]{ieeecolor}

\usepackage{generic}

\usepackage{mathtools}
\usepackage{lipsum}
\usepackage{cite}
\usepackage{amsmath,amssymb,amsfonts}
\usepackage{algorithm,algorithmic}
\usepackage{hyperref}
\hypersetup{hidelinks=true}
\usepackage{textcomp}
\makeatother

\usepackage{graphicx,subfigure}
\usepackage{amsthm}
\usepackage{tabu}
\usepackage{breqn}
\usepackage{verbatim}
\usepackage{caption}
\usepackage{esvect}
\usepackage{footnote}
\usepackage{tikz, pgfplots}

\usepackage{color,xcolor}
\usepackage{url}
\usepackage[makeroom]{cancel}
\usepackage{todonotes}
\usepackage{textcomp}

\newtheorem{theorem}{Theorem}[section]
\newtheorem{corollary}[theorem]{Corollary}
\newtheorem{definition}[theorem]{Definition}
\newtheorem{lemma}[theorem]{Lemma}
\newtheorem{proposition}[theorem]{Proposition}
\newtheorem{remark}{Remark}

\newtheorem{assumption}{Assumption}
\definecolor{myred}{RGB}{108,0,0}
\graphicspath{{./Figures/}}

\DeclareMathOperator*{\argmin}{arg\,min}
\newcommand{\real}{\mathbb{R}}

\newcommand{\ess}{\mathrm{ess}}
\newcommand{\eye}{\mathrm{I}}

\newcommand{\tran}{^{\mkern-1.5mu\mathrm{T}}\!}

\renewcommand{\real}{\mathbb{R}}

\newcommand{\Ac}{{\mathcal{A}}}
\newcommand{\Bc}{{\mathcal{B}}}
\newcommand{\Cc}{{\mathcal{C}}}

\newcommand{\Kc}{{\mathcal{K}}}
\newcommand{\Lc}{{\mathcal{L}}}

\newcommand{\Yc}{\mathcal{Y}}
\newcommand{\Xc}{\mathcal{X}}

\newcommand{\Uc}{\mathcal{U}}

\newcommand{\Hc}{\mathcal{H}}
\newcommand{\Rc}{\mathcal{R}}
\newcommand{\expect}{\mathbb{E}}
\newcommand{\diff}{\mathrm{d}}
\newtheorem{example}{Example}[section]

\renewcommand{\baselinestretch}{.98}

\begin{document}
\title{An Information Theory of Finite Abstractions and their Fundamental Scalability Limits}

\author{\authorblockN{Giannis Delimpaltadakis and Gabriel Gleizer {}}\thanks{Giannis Delimpaltadakis is with the Robust and Intelligent Autonomous Systems lab, Italian Institute of Artificial Intelligence (AI4I), Turin, Italy. Gabriel Gleizer is with the Delft Center for Systems and Control, Mechanical Engineering, Delft University of Technology. Emails: \texttt{ioannis.delimpaltadakis@ai4i.it, g.gleizer@tudelft.nl}. \newline \indent This research is partially supported by the project ``Chaotic sampling for secure and sustainable networks of control systems'' with file number 21937 of the research programme VENI AES 2024 which is (partly) financed by the Dutch Research Council (NWO) under the grant \url{https://doi.org/10.61686/WZNAX74774}.}}   

\maketitle
\begin{abstract}
Finite abstractions are discrete models of dynamical systems, such that the set of abstraction trajectories contains all system trajectories. There is a consensus that abstractions suffer from
a scalability bottleneck: to obtain a sufficient ``accuracy" (how closely the abstraction models the system), the required abstraction size often becomes very large, rendering computation infeasible. This is accentuated for complex, high-dimensional systems, due to the curse of dimensionality. Yet, after decades of research, there are no formal results on the size-accuracy tradeoff of abstractions. Here, we derive a statistical, quantitative theory of the size-accuracy tradeoff of abstractions of autonomous deterministic systems and uncover fundamental limits on their scalability, through \emph{rate-distortion theory}---the information theory of lossy compression. Abstractions are viewed as \emph{encoder-decoder} pairs, encoding trajectories of dynamical systems. \emph{Rate} measures abstraction size, while \emph{distortion} describes accuracy, defined as the spatial average deviation between abstract trajectories and system ones. We obtain a fundamental lower bound on the minimum achievable abstraction distortion, given the system dynamics and the abstraction size; and vice-versa a lower bound on the minimum required size, for given distortion. The bound depends on the complexity of the dynamics, through trajectory entropy and a constant depending on the trajectory manifold's geometry. We demonstrate its tightness on some dynamical systems. Finally, we showcase how this new theory enables constructing minimal abstractions, optimizing the size-accuracy tradeoff, through an example on a chaotic system.
\end{abstract}

\maketitle

\section{Introduction}
Modern engineering systems are becoming more complex and must meet intricate specifications in safety-critical situations. For instance, a self-driving car must follow traffic rules, avoid collisions, and optimize speed and fuel consumption. Due to the complexity of these systems, traditional analytic methods for verification and control are intractable. For over two decades, to address verification and control of complex dynamics and objectives, \emph{abstraction}-based methods have flourished \cite{tabuada2009verification,lavaei2022automated}. Given a dynamical system, these methods construct a finite system---the abstraction---, arising from partitioning the state (and control) space of the original system, such that all trajectories of the system are contained, in a formal sense, in the set of abstraction trajectories. Employing this property, one may solve an intractable verification or control problem for the original system over the finite abstraction, with formal guarantees of correctness. Over the years, research on abstractions has spanned deterministic systems \cite{girard2009approximately,rungger2016scots, mallik2018compositional}, stochastic systems \cite{zamani2014symbolic,lahijanian2015formal, abate2010approximate}, and data-driven scenarios \cite{coppola2023data, badings2023robust, devonport2021symbolic, kazemi2024data}.

Despite their immense success, there is a consensus that abstractions suffer from a scalability bottleneck: to obtain a sufficient ``accuracy" (how closely the abstraction represents the system), the required abstraction size often becomes very large, rendering computation infeasible. This is accentuated for complex, high-dimensional systems, due to the curse of dimensionality induced by partitioning. However, after considerable interest on abstractions in the past decades, there are still no formal results concerning their size-accuracy tradeoff.

\subsection*{Contributions}
We derive a statistical, quantitative theory of the size-accuracy tradeoff of abstractions of autonomous  deterministic systems and uncover fundamental limits on their scalability. To that end, we establish connections with \emph{rate-distortion theory} -- the information theory of lossy compression \cite[Chapter 10]{cover1999elements}. The key observation for the theory is that abstractions are information-theoretic \emph{encoder-decoder} pairs, encoding trajectories of dynamical systems, in a higher-dimensional ambient space. \emph{Rate} measures abstraction size, while \emph{distortion} is defined as the spatial average deviation between abstract trajectories and system ones, thus capturing the average accuracy of an abstraction. Then, building on recent developments in rate-distortion theory for generalized measurable sets \cite{riegler2018rate,riegler2023lossy}, we derive fundamental limits of abstractions' size-accuracy tradeoff: a fundamental lower bound on the minimum achievable abstraction distortion, for given system dynamics and abstraction size; conversely, we also obtain a lower bound on the minimum abstraction size, for given distortion. The  fundamental lower bound depends on the complexity of the dynamics, through generalized entropy, which we show how to compute. We demonstrate the tightness of the bound on certain dynamical systems. Finally, we showcase how the developed theory can be employed to construct optimal abstractions, in terms of the size-accuracy tradeoff, through an example on a chaotic system, and we provide a discussion towards a general procedure for constructing optimal abstractions.
    
\subsection*{Related work}
Through decades of research, there has been considerable effort to construct scalable abstractions. Indicatively, \cite{esmaeil2013adaptive,adams2022formal, tazaki2009discrete, weber2016optimized, saoud2017multirate} adapt the partition's resolution depending on the \emph{local uncertainty} a given state-space region induces to the abstraction. Further, \cite{hsu2018multi} constructs multi-resolution abstractions, employing feedback-refinement relations. The work \cite{calbert2024smart} employs optimal control, such that the generated trajectories result in smaller abstraction cells and only a portion of the state space needs to be partitioned. Although the above methods result in more scalable abstractions, they neither provide quantitative results on the size-accuracy tradeoff, nor optimize some metric describing it. Another approach to derive more accurate abstractions is introducing memory \cite{schmuck2014asynchronous, banse2025memory}, based on sequences of outputs. In \cite{gleizer2022chaos}, it is shown that the size of such memory-based abstractions increases exponentially with the sequence length for deterministic chaotic systems. Apart from adaptive-partitioning techniques, compositional methods \cite{mallik2018compositional,lavaei2020compositional} decompose the system to smaller ones, that are abstracted more efficiently. However, they do not address scalability issues of abstracting each subsystem. Further, it is also worth mentioning \cite{delimpaltadakis2023interval}, which, for a particular class of stochastic abstractions, demonstrates that partitioning the control space may be unnecessary.

The connection between information theory and symbolic dynamics is well-known \cite{lind2021introduction}; listing the whole literature on the topic is impossible. Worth mentioning is the work in \cite{lindenstrauss2018rate}, which employs rate-distortion theory to characterize complexity of dynamical systems and their relationship with so-called \emph{shifts}\footnote{A class of discrete systems. Abstractions can be cast as shifts.}. Nonetheless, \cite{lindenstrauss2018rate} does not consider the deviation between a shift and a dynamical system, but rather focuses on asymptotic results (arbitrarily large partition size, steady-state trajectories) and the qualitative question of if a system can be embedded into a shift. Thus, it does not provide a quantitative theory of the size-accuracy tradeoff. Finally, the works \cite{abel2019state,biza2020learning,larsson2022generalized} employ rate-distortion theory, to compress models that are already discrete and do not focus on abstracting continuous dynamics with formal guarantees.

\section{Preliminaries}

\subsection{Measure spaces, Hausdorff measure, generalized entropy}
We first recall some notions related to information theory over general measurable spaces, based on \cite{riegler2018rate,riegler2023lossy}. When we say ``dimension" of a set, we refer to its intrinsic dimension; e.g. the dimension of a curve in $\real^n$ is 1. We denote the $m$-dimensional Hausdorff measure\footnote{The Hausdorff measure is a generalization of the Lebesgue measure, generalizing the notions of curve length or surface area. E.g., $\Hc^1(\mathcal{C})=2\pi$, where $\Cc$ is the unit circle embedded in $\real^n$.} by $\Hc^m$. Denote the restriction of $\Hc^m$ to the compact set $\Kc$ by $\Hc^m_{\Kc}$. 
Consider a measure $\mu$ over the measure space $(\Xc, \Sigma_\Xc, \nu)$, where $\Sigma_\Xc$ is the Borel $\sigma$-algebra of $\Xc\subseteq\real^n$. When $\mu$ is absolutely continuous w.r.t.~$\nu$ (denoted by $\mu\ll\nu$), we denote the Radon--Nikodym derivative by $\frac{\diff \mu}{\diff \nu}$. When $\nu=\Hc^m_{\Xc}$ (assuming $\Xc$ is $m-$dimensional) and $\mu$ is a probability measure, then $\frac{\diff \mu}{\diff \Hc^m_{\Xc}}$ is the probability distribution associated to $\mu$. Absolute continuity $\mu\ll\Hc^m_{\Xc}$ suggests that $\mu$ is not concentrated in arbitrarily small balls in $\Xc$. We denote the volume of the unit ball in $\real^n$ as $v_n \coloneqq \frac{\pi^{n/2}}{\Gamma(n/2+1)}$, where the Gamma function $\Gamma(a) = \int_{0}^\infty t^{a-1}e^{-t}\diff t$.

Let $\Xc\subset\real^n$ be a finite union of bounded, $m$-dimensional, $C^1$-manifolds. Denote by $c_\Xc\in(0,\infty]$ the smallest constant such that 
\begin{equation}\label{eq:c}
        \Hc^m_{\Xc} (B(\hat{x},\delta)\cap \Xc)\leq c_{\Xc}\delta^m,\quad \text{for all }\hat{x}\in\real^n \ \text{and }\delta >0,
\end{equation}
where $B(\hat{x},\delta)\coloneqq \{x\in\real^n: \ \|x-\hat{x}\| < \delta\}$. The constant $c_{\Xc}$ is finite in many cases of interest, including compact manifolds (see \cite[Lemma 1]{riegler2018rate}).

Consider a random variable $x$, distributed over the measure space $(\Xc, \Sigma_\Xc, \Hc^m_{\Xc})$ with probability measure $\mu_x$. The generalized entropy of $x$ (w.r.t the Hausdorff measure) is
\begin{equation}\label{eq:generalized entropy}
    h(x) \coloneqq -\expect_{x}\!\left[\log\Big(\frac{\diff \mu_x}{\diff \Hc^m_{\Xc}}\Big)\right],
\end{equation}
where $\expect_{x}[\cdot]$ denotes the expectation operator w.r.t. the random variable $x$. The generalized entropy is the extension of the classical Shannon entropy to continuous spaces, and is a measure of uncertainty or complexity of a random variable. For a measure space $(\Xc,\Sigma_{\Xc}, \Hc^m_{\Xc})$ with $0<\Hc^m_{\Xc}(\Xc)<\infty$: a) $h(x)$ is maximized for the uniform distribution $\mu_{x} = \Hc^{m}_{\Xc}/\Hc^{m}_{\Xc}(\Xc)$, b) $h(x)$ is bounded above, when $\mu_x\ll \Hc^m_{\Xc}$. Finally, employing a similar generalization as in \eqref{eq:generalized entropy}, let us denote the \emph{generalized Rényi entropy} with parameter $a\in \real_{+}\setminus\{1\}$ by $h_a(x):=\frac{1}{1-a}\log \expect_x [(\frac{\diff \mu_x}{\diff \Hc_\Xc^m})^{a-1}]$.

The example below shows how the above apply to computing the entropy of a dynamical system's trajectories.

\begin{example}[Entropy of trajectories of the doubling map] \label{ex:doubling map}
    Consider the dynamical system $x^+=f(x)$, where the doubling map $f:[0,1]\to[0,1]:x\mapsto 2x\ \mathrm{mod}\ 1$. Consider the set of $3$-length trajectories of the system $\Bc:=\{(x_0,f(x_0),\allowbreak f(f(x_0))):\ x_0\in [0,1]\} \subseteq [0,1]^3$. Notice that $\Bc$ is the union of 4 straight-line segments: 
\begin{align*}
   \Bc =& \{(x_0,2x_0,4x_0):x_0\in[0,.25]\}\cup\\ &\{(x_0,2x_0,4x_0 - 1):x_0\in[.25,.5]\}\cup\\ &\{(x_0,2x_0-1,4x_0-2):x_0\in[.5,.75]\}\cup\\ &\{(x_0,2x_0-1,4x_0-3):x_0\in[.75,1]\}.
\end{align*}
Consider random initial conditions $x_0\sim U[0,1]$, where $U[0,1]$ is the uniform distribution over $[0,1]$. It is known that $U[0,1]$ is invariant under the doubling map. Thus, the random variable $\xi(x_0) =  (x_0,f(x_0),\allowbreak f(f(x_0)))\in \Bc$, i.e. the trajectories, is uniformly distributed over $\Bc$; i.e., the probability measure $\mu_\xi= \Hc^1_{\Bc}/\Hc^1_{\Bc}(\Bc)$. Its generalized entropy
\begin{align*}
    h(\xi)= -\expect_{\xi}\!\left[\log\Big(\frac{\diff \mu_{\xi}}{\diff \Hc^1_{\Bc}}\Big)\right]&=-\int_{\Bc}\log\Big(\frac{\diff \mu_{\xi}}{\diff \Hc^1_{\Bc}}\Big)\diff \mu_{\xi} \\&= -\int_{\Bc}\log(\frac{1}{\Hc^1_{\Bc}(\Bc)})\diff \mu_{\xi}\\
    &=\log(\sqrt{21})\int_{\Bc}\diff \mu_{\xi}= \log(\sqrt{21}),
\end{align*}
where we have used that the total length of the line segments is $\Hc^1_{\Bc}(\Bc)=\sqrt{21}$, and that $\int_{\Bc}\diff \mu_{\xi}=1$.
\end{example}  

\subsection{Rate-distortion theory on measurable spaces}\label{sec:rd}
\begin{figure}[t]
	    \centering
			\includegraphics[width = \linewidth]{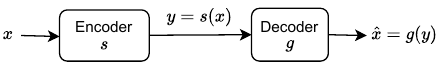}
            \caption{The typical source coding setting.}
			
		\label{fig:rsc}
\end{figure}
A typical setting in information theory is source coding, see Fig. \ref{fig:rsc}. A source emits a \emph{message} $x\in\Xc\subseteq\real^n$, which is a random variable over the measure space $(\Xc,\Sigma_\Xc, \Hc^m_{\Xc})$, with associated probability measure $\mu_x$, where $\Xc$ is assumed to be $m$-dimensional. The \emph{encoder} $s:\Xc\to\Yc$, where $\Yc$ is finite, outputs the coded message $s(x)=y\in\Yc$. Finally, the \emph{decoder} $g:\Yc\to\hat{\Xc}$, upon receiving $y$, decodes it into $g(y)=\hat{x}$. Compression takes place by encoding the continuous message $x$ into a low-dimensional, finite coded message $y$. The encoder cardinality $|\Yc|$ determines the compression, and the \emph{rate} is defined as $\log|\Yc|$. A \emph{distortion} function $d:\Xc\times\hat{\Xc}\to \real_+$ measures the deviation of $\hat{x}$ from the original message $x$. A typical distortion function, when $\hat{\Xc}=\real^n$, is the squared error $d(x,\hat{x}) = \|x-\hat{x}\|^2$.\footnote{Here, we present a simplified setting of source coding, where the encoder-decoder is deterministic, and the source emits a single message. For the general theory, see \cite{cover1999elements,riegler2023lossy}.}   

Of particular interest is the fundamental limit of the rate-distortion tradeoff, i.e. the following quantity:
\begin{equation*}
\begin{aligned}
    D(R) := \inf_{s,g}& \ \expect_{x} [d(x,\hat{x})\mid s,g]\\
    \mathrm{s.t.}& \ s:\Xc\to\Yc, \ g:\Yc\to \hat{\Xc},\\
    &\ \log|\Yc|\leq R, \ y= s(x), \ \hat{x} = g(y),
\end{aligned}
\end{equation*}
where the expectation is taken w.r.t. the random variable $x$. In words, $D(R)$ is the minimum achievable average distortion, for a given rate threshold $R$. The function $D(R)$ has an inverse, $R(D)$, which is the minimum rate, for a given maximum expected distortion threshold $D$. The following result provides a fundamental lower bound on $D(R)$.

\begin{theorem}[Generalized Shannon lower bound {\cite[Thm. 3.1, simplified]{riegler2023lossy}}]\label{thm:slb}
    Let $\Xc\subseteq\real^n$ be a finite union of compact, $m$-dimensional, $C^1$-manifolds, and $\mu_x \ll \Hc^m_{\Xc}$. Assume that $\hat{\Xc}\subseteq\real^n$ and that $(\hat{\Xc},\Sigma_{\hat{\Xc}})$ is measurable. Consider the Euclidean distortion $d:\Xc\times\hat{\Xc}\to \real_+:(x,\hat{x})\mapsto \|x-\hat{x}\|^2$. Then
    \begin{align}
    D(R)&\geq D_*(R)\coloneqq \frac{m}{2}\Big(\frac{e^{-R+h(x)-m/2}}{c_{\Xc}\Gamma(1+m/2)}\Big)^{2/m}.
    \end{align}
\end{theorem}
\begin{proof}[Proof Sketch]
    This is \cite[Thm. 3.1]{riegler2023lossy} specialized for finite unions of compact $C^1$-manifolds and Euclidean distortion. 
\end{proof}

\subsection{Transition systems}
\begin{definition}[Transition system]
    A transition system $S$ is a tuple $S=(\Xc,\underset{S}{\rightarrow})$, where $\Xc$ is the state space and $\underset{S}{\rightarrow}\subseteq\Xc\times\Xc$ is a transition relation.
\end{definition}
A transition system $S=(\Xc,\underset{S}{\rightarrow})$ is \emph{deterministic} if, for any $x\in\Xc$, there exists at most one $x'\in\Xc$, such that $(x,x')\in\underset{S}{\rightarrow}$.  Given a transition system $S=(\Xc,\underset{S}{\rightarrow})$, its \emph{$l$-length behavior} $\Bc_l^S$ is defined as $\Bc_l^S :=\Big\{\xi: \ \xi = \{x_i\}_{i=0}^{l-1}, \ (x_{i},x_{i+1})\in\underset{S}{\rightarrow}, \ i=0,1,\dots,l-2\Big\}$. That is, the $l$-length behavior is the set of $l$-long trajectories. Notice that $\Bc_l^S\subseteq \Xc^l$. 

\section{Abstractions and the curse of dimensionality}\label{sec:abstractions}

\subsection{Finite abstractions of dynamical systems}

Throughout this work, we consider deterministic dynamical systems $x^+ = f(x)$, with $f:\Xc\to\Xc$. Dynamical systems obtain the transition-system representation $S=(\Xc, \underset{S}{\rightarrow})$, where $\underset{S}{\rightarrow} := \{(x,y):y=f(x), \ x\in\Xc\}$. We make the following assumption.
\begin{assumption}[The state space]\label{assum:X}
    The set $\Xc\subseteq\real^n$ is $n$-dimensional, connected and compact.
\end{assumption}
Under Assumption \ref{assum:X}, $\Bc_l^S$ is an $n$-dimensional subset of $\Xc^l\subseteq\real^{nl}$. We now introduce \emph{abstractions} of dynamical systems.
\begin{definition}[Measurable Partition]
    Given a set $\Xc$, a finite collection of measurable, disjoint sets $\Yc=\{Y_i\}$, such that $\bigcup_i Y_i = \Xc$, is a \emph{measurable partition} of $\Xc$.
\end{definition}
 
\begin{definition}[Abstraction]\label{def:abstraction}
    Given a dynamical system with transition-system representation $S=(\Xc, \underset{S}{\rightarrow})$ and a measurable partition $\Yc=\{Y_i\}$ of $\Xc$, a transition system $ A=(\Yc, \underset{A}{\rightarrow}\nolinebreak)$ is an abstraction of $S$ if, for any $x,x'\in\Xc$ and $Y,Y'\in \Yc$, such that $x\in Y$ and $x'\in Y'$, we have $(x,x')\in\underset{S}{\rightarrow}$ $\implies$ $(Y,Y')\in \underset{A}{\rightarrow}$.
\end{definition}
We consider abstractions based on state-space partitions for ease of exposition, but our results apply to more general notions, such as \emph{simulation-based abstractions} (\cite[Def. 9.2]{tabuada2009verification} specialized to autonomous systems), and \emph{finite covers of the system's trajectory space}. More detail is given in Remark \ref{rem:extension_details}.

Although the dynamical system $S$ is deterministic, the abstraction $A$ is generally non-deterministic. With a slight abuse of formality, we often treat trajectories $\{\omega_i\}_{i=0}^{l-1}$ of the abstraction (with $\omega_i\in\Yc$) as subsets of $\Xc^l$, that is $\{\omega_i\}_{i=0}^{l-1} \equiv \omega_0\times \omega_1 \times \dots \times \omega_{l-1} \subseteq \Xc^l$. We will also make use of the set of $l$-long reachable sets of the abstraction $A$, that is $\Rc_l^A\coloneqq \{\Omega\in 2^{\Xc^{l}}:\ \Omega =\bigcup\limits_{\omega\in\Bc_l^A,\ \omega_0=Y}\omega,\ Y\in \Yc\}$. Observe that the cardinality $|\Rc_l^A|=|\Yc|$.

\begin{theorem}[Behavioral inclusion {\cite[Theorem 4.18, simplified]{tabuada2009verification}}]\label{thm:inclusion}
    For a dynamical system $x^+=f(x)$ with transition system representation $S=(\Xc,\underset{S}{\rightarrow})$, a partition $\Yc$ of $\Xc$ and an abstraction $A = (\Yc, \underset{A}{\rightarrow})$ of $S$, the following holds for any $l$: for any $\xi\in \Bc_l^S$, there exists $\omega\in\Bc_l^A$ such that $\xi\in\omega$.
\end{theorem}
An immediate result of Thm. \ref{thm:inclusion} is the following:
\begin{corollary}\label{cor:inclusion}
    For a dynamical system $x^+=f(x)$ with transition system representation $S=(\Xc,\underset{S}{\rightarrow})$, a partition $\Yc$ of $\Xc$ and an abstraction $A = (\Yc, \underset{A}{\rightarrow})$ of $S$, the set $\Rc_l^A$ is a finite cover of $\Bc_l^S$.
\end{corollary}
\begin{proof}
    Follows directly from Thm. \ref{thm:inclusion} and the fact that $|\Rc_l^A|=|\Yc|$.
\end{proof}
This observation is instrumental in this work. Through behavioral inclusion, abstractions encode information about the infinite, continuous system behavior $\Bc_S^l$ into the finite set $\Rc_{l}^A$. While this enables computational methods to verification problems for dynamical systems, it also generally entails information loss, as the following section explains.

\subsection{Abstraction-based verification and information loss}\label{subsec:verification}
In typical verification problems, we are given a set of initial conditions $\Xi_0\subseteq \Xc$ for the system $S$ and we have to check if the corresponding set of system trajectories $\Xi=\{\xi \in \Bc_l^S: \ \xi_0\in \Xi_0\}$ satisfies a given property. For example, in the case of safety, we have to check if $\Xi\cap\Uc_l = \emptyset$, where $
\Uc_l:=\{\xi\in \Xc^l:\exists i\in\{0,\ldots,l-1\},\ \xi_i\in \Uc\}$ and $\Uc\subseteq\Xc$ is an unsafe set. Computing the exact reachable set $\Xi$ is generally impossible. Abstractions $A$ address this problem by computing the corresponding set of abstract state trajectories $\Omega_A = \bigcup\limits_{\omega\in \Bc_l^A, \ \omega_0 \cap \Xi_0 \neq \emptyset}\omega$, which is tractable, as the abstraction is finite. Notice that, by Thm. \ref{thm:inclusion}, we have $\Xi\subseteq \Omega_A$. Hence, e.g. for safety verification if $\Omega_A\cap \Uc_l=\emptyset$, then one may safely deduce that the system is safe, as $\Omega_A\cap \Uc_l=\emptyset\implies \Xi\cap\Uc_l = \emptyset$. 

As abstractions group system states $x\in\Xc$ in sets $Y\subseteq\Xc$, \emph{information loss} is inevitable. In general, the partition $\Yc$ needs to have a relatively high resolution, to recover a meaningful verification answer. E.g., in the extreme case of $|\Yc|=1$, for any set of initial conditions $\Xi_0\subseteq \Xc$, the abstraction returns $\Omega_A=\Xc^l$, i.e. the whole ambient space of $l$-length trajectories. As such, for small $|\Yc|$, the abstraction $A$ does not accurately represent the system $S$. On the other hand, for large $|\Yc|$, where the abstraction is more accurate, the computations on the abstraction become heavier -- even intractable. Thus, there is a trade-off between abstraction accuracy and partition size $|\Yc|$. In what follows, we provide a statistical, quantitative theory of this size-accuracy tradeoff and provide bounds on it, based on rate-distortion theory. 

\section{Information-theoretic framework for finite abstractions}
\begin{figure*}
    \centering
    \includegraphics[width=.85\linewidth]{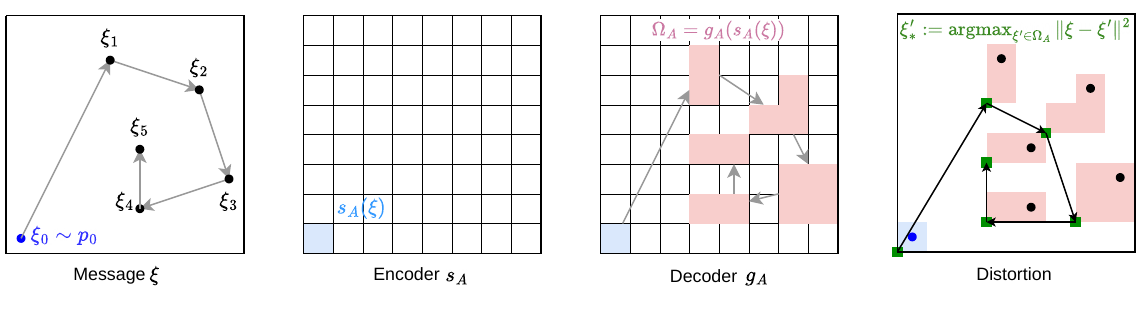}
    \caption{Abstractions as a source coding scheme. From the left: 1) a sample trajectory $\xi$ of system $S$ with its initial state $\xi_0$ highlighted in blue; 2) the state-space partition $\Yc$, and the corresponding abstract initial condition in cyan; 3) the set of abstract trajectories $\Omega_A$, in red; 4) the specific trajectory $\xi'_*\in\Omega_A$ that deviates the most from the true system trajectory $\xi$.}
    \label{fig:abstractions_as_rsc}
\end{figure*}
In what follows, consider the dynamical system $x^+=f(x)$, with $x\in \Xc\subseteq\real^n$, under Assumption \ref{assum:X}. The dynamical system admits the transition system representation $S=(\Xc, \underset{S}{\rightarrow})$. Towards deriving a statistical quantification on the size-accuracy tradeoff of abstractions, we impose a probability distribution $p_{\xi_0}:\Xc\to\real_+$ on the system's initial conditions. Verification, then, becomes: sampling an initial condition $\xi_0\in \Xc$, with $\xi_0\sim p_{\xi_0}$, and afterward employing the abstraction to give a verification answer.\footnote{More precisely, $\xi_0$ is a random variable over $(\Xc,\Sigma_{\Xc}, \Lc^n_{\Xc})$, where $\Lc^n$ is the $n$-dimensional Lebesgue measure, with probability measure $\mu_{\xi_0}$ such that $\frac{\diff \mu_{\xi_0}}{\diff \Lc^n_{\Xc}}=p_{\xi_0}$.}

Let us show how an abstraction $A$ can be viewed as an encoder-decoder pair of system trajectories $\xi \in \Bc_l^S$. For the following, we refer the reader to Figure \ref{fig:abstractions_as_rsc}. The system (source) samples an initial condition $\xi_0\sim p_{\xi_0}$ and generates the trajectory $\xi=(\xi_0,\dots,\xi_{l-1})\in \Bc_l^S$ (the message). The encoder $s_A:\Bc_l^S\to\Yc$ looks at the initial condition $\xi_0$ and returns the corresponding abstract initial condition:\footnote{As explained in Section \ref{sec:abstractions}, for verification, the only information given to the abstraction is the initial condition $\xi_0$. Nonetheless, for given dynamics $f$,  $\xi_0$ and $\xi$ are equivalent and carry the exact same information, as $\xi_0\mapsto \xi$ is bijective. Thus using $\xi$ as the encoder input is equivalent to using $\xi_0$. We adopt $\xi$ to regard $\Bc_l^S$ as the message space, which is convenient since the distortion \eqref{eq:abstraction_distortion_function} is defined over trajectories (see also Remark \ref{rem:message space}).}
\begin{equation}\label{eq:abstraction_encoder}
    s_A(\xi):= Y, \ \text{s.t.} \ \xi_0\in Y
\end{equation}
The decoder $g_A$, upon receiving the initial condition $s_A(\xi)$, outputs the set of all abstract state trajectories corresponding to $s_A(\xi)$. That is, for the decoder we have $g_A:\Yc\to \Rc_l^A$ with
\begin{equation}\label{eq:abstraction_decoder}
    g_A(y) :=\bigcup\limits_{\omega\in \Bc_l^A, \ \omega_0=y}\hspace{-3mm}\omega.
\end{equation}
The rate, determined by the encoder's size, is $\log(|\Yc|)$. Indeed, notice that \emph{the abstraction encodes the system's trajectories $\Bc_l^S$ into exactly $|\Yc|$ outcomes, that is $\{g_A(z): \ z\in \Yc\}=\Rc_l^A$.} 

To capture the accuracy of the abstraction, and compare the message $\xi$ and output $\Omega_A = g_A(s_A(\xi))$, we employ a distortion function $d:\Bc_l^S\times2^{\Xc^l}\to \real^+$
defined by
\begin{equation}\label{eq:abstraction_distortion_function}
    d(\xi,\Omega_A):= \sup_{ \xi' \in \Omega_A} \frac{1}{l}\|\xi-\xi'\|^2.
\end{equation}
In words, $d(\xi,\Omega_A)$ is the worst possible distortion between system trajectories $\xi$ and abstract trajectories $\Omega_A$, averaged over the time horizon $l$. This is in-line with abstraction-based verification, where the worst-case outcome is considered.

Let us now explain what ``expected (or average) distortion", for a given abstraction $A$, means in the context of verification. The expected distortion $\expect_{\xi_0}[d(\xi,\Omega_A)\mid A]$ is taken w.r.t. the initial-condition distribution $p_{\xi_0}$. Thus, for $N$ verification problems with $N\to\infty$, where the initial condition $\xi_0\sim p_{\xi_0}$, $\expect_{\xi_0}[d(\xi,\Omega_A)\mid A]$ is the average distortion. As $d$ measures the distance between system trajectories and abstract state trajectories, the expected distortion $\expect_{\xi_0}[d(\xi,\Omega_A)\mid A]$ is thus the \emph{spatial}, statistical average of the deviation between system trajectories and abstract state trajectories, over initial conditions in $\Xc$ with distribution $p_{\xi_0}$. 

\begin{remark}[Initial-condition distribution]
The distribution $p_{\xi_0}$ weighs how much each initial condition $\xi_0\in\Xc$ (and thus each trajectory $\xi\in\Bc_l^S$) contributes to the average distortion $\expect_{\xi_0} [d(\xi,\Omega_A)\mid A]$. Arguably, a suitable choice for $p_{\xi_0}$ is the uniform distribution, as when constructing an abstraction, the initial condition is unknown and all initial conditions are considered equally likely; however, if the system is known to obey a specific distribution $p_{\xi_0}$, we may use this one.
\end{remark}

Finally, the optimal abstraction size-accuracy tradeoff is captured by the following rate-distortion quantity:
\begin{align*}
    D_{abs}(R) := \inf_{A}& \ \expect_{\xi_0} [d(\xi,\Omega_A)\mid A]\\
    \mathrm{s.t.}& \ A\text{ is an abstraction of }S,\\
    &\ \text{\eqref{eq:abstraction_encoder}, \eqref{eq:abstraction_decoder} hold}, \\
    &\ \log|\Yc|\leq R, \ \Omega_A = g_A(s_A(\xi)).
\end{align*}
That is, the minimum average deviation of abstract state trajectories and system trajectories, over all abstractions with a given upper-bound $e^R$ on partition size. Likewise, we also consider the inverse $R_{abs}(D)$, which is the ($\log$ of the) minimum partition size for a given upper threshold $D$ on the average deviation of abstract state trajectories and system ones.

\begin{remark}[Statistics of abstractions' accuracy and size] The proposed theory does not aim at providing (probabilistic) guarantees on the correctness of abstractions. These are a-priori provided by Definition \ref{def:abstraction}, through behavioral inclusion or related properties. Instead, the theory developed here provides (guarantees on the) statistical quantification of abstractions' accuracy and size.
\end{remark}
\begin{remark}[The message space is $\Bc_l^S$]\label{rem:message space}
    Even though we have reduced everything thus far to the initial condition distribution $p_{\xi_0}$, the message space is $\Bc_l^S$, i.e. the system trajectories. Indeed, although the expectation $\expect [d(\xi,\Omega_A)\mid A]$ can be taken either w.r.t. $\xi_0\sim p_{\xi_0}$ or w.r.t. the random variable $\xi\in \Bc_l^S$ (as $\xi_0\mapsto \xi$ is bijective), the distortion $d$ considers the whole $\xi\in\Bc_l^S$. As such, in the coming section, to derive bounds on $D_{abs}(R)$ and $R_{abs}(D)$, employing the theory from Section \ref{sec:rd}, we reason about the random variable $\xi\in\Bc_l^S$ and its probability measure $\mu_{\xi}$ over $(\Bc_l^S, \Sigma_{\Bc_l^S}, \Hc^n_{\Bc_l^S})$, which is solely determined by the initial condition distribution $p_{\xi_0}$ and the system dynamics $f$. Hence, we take expectations $\expect_{\xi}$ and $\expect_{\xi_0}$ interchangeably.
\end{remark}

\section{Rate-distortion theory and a fundamental limit for abstractions}\label{sec:main}

\subsection{A fundamental limit on abstracting dynamical systems}

Having modeled the statistics of abstraction-based verification as a source coding problem, we now proceed to probing the fundamental limits of abstractions' size-accuracy tradeoff, by providing lower bounds on $R_{abs}(D)$ and $D_{abs}(R)$.

Note that abstractions, given the message, output sets and the associated distortion \eqref{eq:abstraction_distortion_function} is set-based. This is in contrast to typical encoder-decoder pairs considered in Thm.~\ref{thm:slb}, which output points and the distortion function is the Euclidean distance. Thus, the results from Section \ref{sec:rd} do not straightforwardly apply, to derive bounds on $D_{abs}$ and $R_{abs}$. In what follows, we derive said bounds, both employing Thm.~\ref{thm:slb} and quantifying the aforementioned distortion disparity. This enables a rate-distortion theory for abstractions. First, we present an intermediate, purely geometric result, providing a lower bound on the average distortion of a given abstraction.
\begin{proposition}[Abstraction vs. encoder distortion]\label{prop:abstraction vs encoder distortion}
    Consider a dynamical system $x^+=f(x)$ with transition system representation $S=(\Xc,\underset{S}{\rightarrow})$. Let Assumption \ref{assum:X} hold. Let $\xi\in \Bc_l^S$ be a trajectory of $S$, with $\xi_0\sim p_{\xi_0}:\Xc\to\real_+$. Consider a measurable partition $\Yc$ of $\Xc$ and an associated abstraction $A$, and let $\Omega_A = g_A(s_A(\xi))$, where $s,g$ are given by \eqref{eq:abstraction_encoder} and \eqref{eq:abstraction_decoder}. Consider an encoder-decoder pair $(s_{q_A},g_{q_A})$, where $s_{q_A}(\xi) = g_A(s_A(\xi))$ and $g_{q_A}(z) = x_c(z)$, where $x_c(z):= \argmin_{y}\sup_{y'\in z}\|y-y'\|^2$ is the Chebyshev center of the set $z$. Denote the Chebyshev radius of $z$, by $r_c(z):=\min_{y}\sup_{y'\in z}\|y-y'\|$. Let $\xi_{q_A} = g_{q_A}(s_{q_A}(\xi))$. The following lower holds for the average distortion of the abstraction:
    \begin{equation}\label{eq:prop abstraction vs encoder}
        \expect_{\xi_0}[d(\xi,\Omega_A)] \geq \frac{1}{l}\expect_{\xi_0}[\|\xi-\xi_{q_A}\|^2] + \frac{1}{l}\expect_{\xi_0}[r_c^2(\Omega_A)],
    \end{equation}
    where $d$ is the distortion function in \eqref{eq:abstraction_distortion_function}.
\end{proposition}
Prop. \ref{prop:abstraction vs encoder distortion} suggests that the average distortion of an abstraction is lower bounded by the expected distortion of a particular encoder-decoder pair (the one outputting the Chebyshev centers of the abstractions outputs) plus a term depending on the size of the abstraction's outputs. Employing Prop. \ref{prop:abstraction vs encoder distortion}, in Theorem \ref{thm:slb_abstractions} below, we derive a fundamental lower bound on $D_{abs}(R)$, by lower-bounding each of the two terms in the right-hand side of \eqref{eq:prop abstraction vs encoder} separately, over all abstractions with the same rate. The first term in the right-hand side of \eqref{eq:prop abstraction vs encoder} can be lower bounded as in Thm.~\ref{thm:slb}, being the expected distortion of an encoder-decoder pair with the same rate as the abstraction. To bound the second term, we observe that the abstraction's outputs $\Omega_A$ form the $nl$-dimensional cover $\Rc_l^A$ of $\Bc_l^S$, and the cover's size is equal to the abstraction's size; the bound is then obtained by lower-bounding over all possible $nl$-dimensional covers of $\Bc_l^S$, using geometric measure theory (see Lemma \ref{lem:sphere_packing}).
For an illustrative example of the above, see Fig. \ref{fig:cheby}.
\begin{figure}[htb]
	    \centering
			\includegraphics[width = 0.6\linewidth]{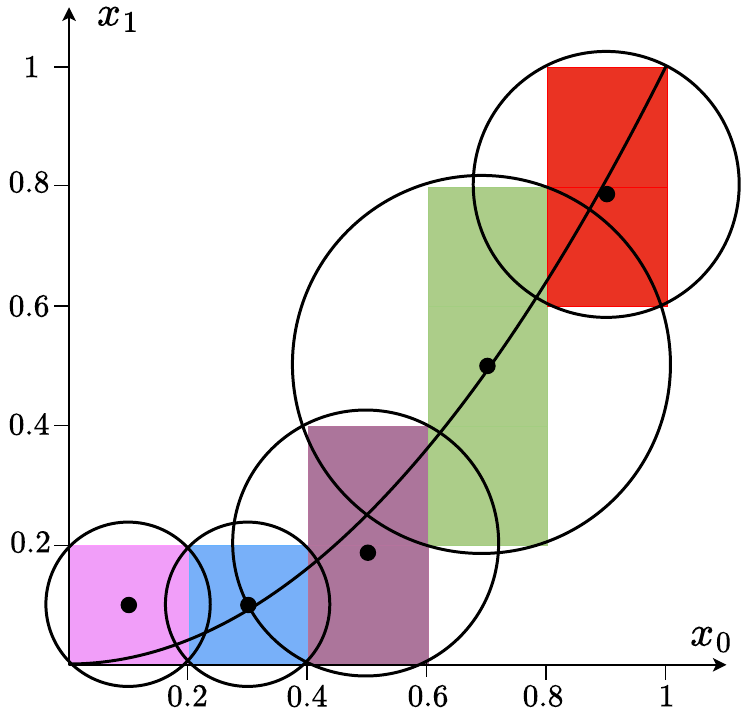}
            \caption{Consider the dynamical system $x^+ = x^2$ with state-space $\Xc=[0,1]$. The parabola depicts the set of trajectories $\Bc_2^S$, embedded in $[0,1]^2$. Consider an abstraction $A$ with associated partition sets $Y_i = [0.2(i-1), \,0.2i)$ for $i=1,\dots,4$ and $Y_5=[0.8,1]$. The abstraction transitions, thus, are $\underset{A}{\rightarrow}=\{(Y_1,Y_1),(Y_2,Y_1), (Y_3,Y_1), (Y_3,Y_2), (Y_4,Y_2), (Y_4,Y_3),\allowbreak (Y_4,Y_4), (Y_5, Y_4), (Y_5,Y_5)\}$. The colored rectangles represent the abstraction outputs $\Omega_A$, depending on the initial condition $\xi_0$. For example, if $\xi_0\in Y_4$, then the output $\Omega_A$ is the green rectangle $\{(x_0,x_1):\, x_0\in Y_4, x_1 \in Y_2\cup Y_3\cup Y_4\}$. Observe that the abstraction's outputs form a $2$-dimensional cover of $\Bc_2^S$, as per Cor. \ref{cor:inclusion}. The dots represent the Chebyshev centers for each different output $\Omega_A$, and the circles are the corresponding Chebyshev balls. E.g, when $\xi_0\in Y_4$, we have $x_c(\Omega_A) = (.7\, \, .5)$ and $r_c(\Omega_A)=\sqrt{0.1}$. The abstraction's average distortion is lower bounded as in \eqref{eq:prop abstraction vs encoder}.
            }
            \label{fig:cheby}
\end{figure}
\begin{theorem}[Shannon lower bound for abstractions]\label{thm:slb_abstractions}
    Consider a dynamical system $x^+=f(x)$ with transition-system representation $S=(\Xc,\underset{S}{\rightarrow})$, and let Assumption \ref{assum:X} hold. Let $\xi\in \Bc_l^S$ be a trajectory of $S$, with $\xi_0\sim p_{\xi_0}:\Xc\to\real_+$. Assume that: 
    \begin{enumerate}
        \item $\Bc^S_l$ is a finite union of bounded, $n$-dimensional $C^1$-manifolds,
        \item $\mu_{\xi}\ll\Hc^n_{\Bc_l^S}$, with $p_\xi\coloneqq \frac{\diff \mu_{\xi}}{\diff \Hc^n_{\Bc_l^S}}$.
    \end{enumerate}
     The average distortion of any abstraction $A$ with partition size $|\Yc|\leq e^R$, where $R>0$, is lower bounded as follows 
    \begin{equation}\label{eq:slb_abstraction_distortion}
    \begin{aligned}
        D_{abs}(R)  \geq\, &\frac{n}{2l}\Big(\frac{e^{-R+h(\xi)-n/2}}{c_{\Bc_l^S}\Gamma(1+n/2)}\Big)^{2/n} \\&+\frac{1}{l} c_{\Bc_l^S}^{-2/n} \sup_{s\in (1,\infty]} e^{\frac{2}{n}(-\frac{s}{s-1}R+h_{s}(\xi))},
    \end{aligned}
    \end{equation}
    where $c_{\Bc_l^S}$ is defined by \eqref{eq:c}.
\end{theorem}
Notice that a valid lower bound on $D_{abs}(R)$ is obtained for any value of $s\in(1,\infty]$ in the right-hand side of \eqref{eq:slb_abstraction_distortion}; the $\sup$ over $s$ provides the tightest bound. In the numerical examples in Section \ref{sec:numerics}, we compute the bound for multiple values of $s$. Further, one may recover a lower bound on $R_{abs}(D)$ numerically, by fixing $D$ in the left-hand side of \eqref{eq:slb_abstraction_distortion} and solving numerically for $R$ (as the right-hand side is a decreasing function of $R$, this is trivially computed by, e.g., bisection methods). Nonetheless, by specifically picking $s=\infty$, one may analytically invert the bound \eqref{eq:slb_abstraction_distortion} to obtain a lower bound on $R_{abs}(D)$, as the following corollary shows.
\begin{corollary}[to Thm. \ref{thm:slb_abstractions}]\label{cor:slb_abstractions}
    Under the assumptions of Thm.~\ref{thm:slb_abstractions}, the minimum abstraction size, for an upper threshold $D$ on average abstraction distortion, is lower bounded as:
    \begin{equation}\label{eq:slb_abstraction_rate}
    \begin{aligned}
        R_{abs}(D)  \geq & \frac{n}{2}\log\!\Big(\frac{n}{2lc_{\Bc_l^S}^{2/n}}\frac{e^{\frac{2}{n}h(\xi)-1}}{\Gamma(1+n/2)^{2/n}} + \frac{1}{lc_{\Bc_l^S}^{2/n}}e^{\frac{2}{n}h_\infty(\xi)}\Big) \\&-\frac{n}{2}\log(D).
    \end{aligned}
    \end{equation}
\end{corollary}
\begin{proof}[Proof Sketch]
    Pick $s=\infty$ in \eqref{eq:slb_abstraction_distortion}, fix the right-hand side equal to $D$, and solve for $R$.
\end{proof}
Hence, Thm~\ref{thm:slb_abstractions} and Cor.~\ref{cor:slb_abstractions} provide fundamental limits on the size-accuracy tradeoff, or the scalability, of abstractions, for given dynamics $x^+=f(x)$.

\begin{remark}[On the assumptions of Thm. \ref{thm:slb_abstractions}]
    The first assumption of Thm. \ref{thm:slb_abstractions}, requiring $\Bc^S_l$ to be a union of smooth manifolds, is satisfied whenever the dynamics $f$ is piecewise continuously-differentiable. The second assumption is satisfied whenever $f$ is piecewise continuous and the initial condition distribution $p_{\xi_0}$ is such that $\mu_{\xi_0}\ll \Lc^n$, where $\Lc^n$ the Lebesgue measure in $\real^n$. 
\end{remark}

\begin{remark}[Simulation-based abstractions and trajectory covers]\label{rem:extension_details}
    Our results apply to abstractions based on $\varepsilon$-approximate simulation relations \cite[Def. 9.2]{tabuada2009verification}, with one minor modification. In such abstractions, the state space is a cover of $\Xc$, and each initial condition $\xi_0$ may be related to multiple $Y_i\in\Yc$. For verification, though, any such initial condition $Y_i$ can be picked, and the resulting reachable set $\Omega_A$ is guaranteed to contain $\xi$ (by the simulation-relation definition). Thus, when there are multiple $Y_i\in \Yc$ such that $\xi_0\in Y_i$, the encoder
    in \eqref{eq:abstraction_encoder} picks only one arbitrarily and deterministically. The reader can verify that all results apply directly, and the resulting bounds are the same. More importantly, our results apply to any finite cover $\Rc_l^A$ of the system's behavior $\Bc_l^S$. As  simulation-based abstractions induce finite covers $\Rc_l^A$ of $\Bc_l^S$ (Cor. \ref{cor:inclusion}), trajectory covers generalize abstractions. Research on such a generalization is left for future work. 
\end{remark}
In the coming section, we provide: a) closed-form expressions on $h(\xi)$, $h_s(\xi)$ and $c_{\Bc_l^S}$, for certain classes of dynamics $x^+=f(x)$, and b) an interpretation on how the complexity of the dynamics and the time-horizon $l$ affect the fundamental lower bound \eqref{eq:slb_abstraction_distortion} in Thm. \ref{thm:slb_abstractions}. Afterwards, in Section \ref{ssec:curse}, we show how the curse of dimensionality arises, through a relaxed version of the bounds of Thm.~\ref{thm:slb_abstractions} and Cor.~\ref{cor:slb_abstractions}.

\subsection{Interpretation and calculus for Theorem \ref{thm:slb_abstractions}}
Before we proceed with the interpretation of Thm. \ref{thm:slb_abstractions}, let us show how one may compute $h(\xi)$, $h_s(\xi)$ and $c_{\Bc_l^S}$, which are required to compute the lower bound \eqref{eq:slb_abstraction_distortion}. Let us define the function $b_l:\Xc\to\Bc_l^S$ by 
\begin{equation}\label{eq:b}
    b_l(x) \coloneqq \big[\,\begin{matrix}
    x\tran & f(x)\tran & f(f(x))\tran & \cdots & f^{(l-1)\tran}(x) 
\end{matrix}\,\big]\tran,
\end{equation}
which maps an initial state into its $l$-long trajectory. 

\begin{proposition}[Computing $h(\xi)$ and $h_s(\xi)$]\label{prop:entropy}
    Consider a system $x^+ = f(x)$ with $f : \Xc \to \Xc$ measurable and piecewise Lipschitz
    \footnote{That is, $\Xc$ is a finite union $\cup_i\Xc_i$ of Lebesgue-measurable sets such that the restriction of $f$ to each $\Xc_i$ is Lipschitz. This condition may be relaxed to $f$ \emph{approximately Lipschitz}, see \cite[Thm.~3.1.8, Sec.~3.2.1]{federer1969geometric}, which also implies approximate differentiability.} 
    and differentiable. Let Assumption \ref{assum:X} hold, and $\xi_0 \sim p_{\xi_0}:\Xc\to\real_+$. The following expressions hold
    \begin{equation}\label{eq:trajectory entropy}
    h(\xi) = h(\xi_0) + \frac{1}{2}\int_\Xc p_{\xi_0}(x)\log\det(J_{b_l}(x)\tran J_{b_l}(x))\diff x,
    \end{equation}
    \begin{equation}\label{eq:renyi trajectory}
        h_{s}(\xi) = \frac{1}{1-s} \log \int_{\real^n} \frac{p_{\xi_0}(x)^{s}}{\det(J_{b_l}(x)\tran J_{b_l}(x))^{\frac{s-1}{2}}} \, \diff x, \quad s>1,
    \end{equation}
    \begin{equation}\label{eq:renyi max trajectory}
        h_{\infty}(\xi) = \ess\,\inf_\Xc\Big(\frac{1}{2}\log\det(J_{b_l}(x)\tran J_{b_l}(x)) - \log p_{\xi_0}(x)\Big),
    \end{equation}
    where $J_{b_l}$ denotes the Jacobian matrix of $b_l$. Moreover, $h(\xi) \geq h(\xi_0)$ and $h_s(\xi) \geq h_s(\xi_0)$ for any $s \in [1,\infty).$
\end{proposition}

\begin{proposition}[Computing $c_{\Bc_l^S}$]\label{prop:c}
    Consider a system $x^+ = f(x)$, with $f : \Xc \to \Xc$ differentiable a.e., and let Assumption \ref{assum:X} hold. The following facts on $c_{\Bc_l^S}$ hold:
    \begin{enumerate}
        \item $c_{\Bc_l^S} \leq v_n$, if $f$ is affine;
        \item $c_{\Bc_l^S} \leq M^lv_n$, if $f$ is piecewise affine with $M$ modes;
        \item $c_{\Bc_l^S} \leq v_n\big(\sum_{i=0}^{l-1} L^{2i}\big)^{n/2}$, if $f$ is Lipschitz continuous with constant $L$.
    \end{enumerate}
\end{proposition}

\begin{remark}[$c_{\Bc_l^S}$ at high rates]\label{rem:c_piecewise} As the partition size $|\Yc|$ grows large, the Chebyshev balls of the abstraction outputs (c.f. Prop. \ref{prop:abstraction vs encoder distortion} and Lemma \ref{lem:sphere_packing}) become small. Hence, in the case of smooth $f$, their intersection with the manifold $\Bc_l^S$ approaches the case of an affine system, with $c_{\Bc_l^S} \leq v_n.$ Similarly, for piecewise affine $f$ and sufficiently small balls -- at least $|\Yc| \geq M^l$ --, these can be chosen to intersect with at most one piece each. Thus, in the high-rate regime, to reduce conservatism of the bound, one may use the approximation  $c_{\Bc_l^S} \approx v_n.$ We demonstrate this in the numerical examples in Section \ref{sec:numerics}.    
\end{remark}

We proceed to discussing Thm. \ref{thm:slb_abstractions}. First, inspecting  \eqref{eq:slb_abstraction_distortion}, systems with more complex dynamics lead to bigger abstraction distortion, for fixed abstraction size, since the right-hand side is increasing w.r.t. $h(\xi)$ and the Rényi entropy $h_s(\xi)$; equivalently, more complex systems require bigger abstraction size for the same distortion.

Regarding the effect of the time-horizon $l$ on the bound \eqref{eq:slb_abstraction_distortion}, we have to inspect the effect that $l$ has on $h(\xi)$ and $h_s(\xi)$. Let us first demonstrate that, for the ``simple" dynamics of exponentially stable systems, the abstraction distortion can be made arbitrarily small for $l\to\infty$, for sufficiently large abstraction size.

\begin{example}[Exponentially stable systems] Consider a system $x^+ = f(x)$ whose origin is exponentially stable on a given compact set in $\real^n$. Then, there is a Lyapunov function $V: \real^n \mapsto \real_+$ satisfying $V(x) \geq \frac{1}{r^2}\|x\|^2$ for a given $r > 0$ and, for all $x$ s.t.~$V(x)\leq 1$ (w.l.o.g.), $V(f(x)) \leq aV(x),$ with $ a \in [0, 1).$ This allows us to create an abstraction $A$ with the associated partition $Y_i = \{x \in \real^n \mid a^i < V(x) \leq a^{i-1}\}$ for $i = 1,...,N,$ and $Y_{N+1} = \{x \in \real^n \mid V(x) \leq a^N\}$; and transitions $Y_i \xrightarrow[A]{} Y_j$  if and only if $j > i$ or $j = i = N+1$. The abstraction encapsulates the fact that, after $N$ steps or less, all trajectories reach the sublevel set $V(x) \leq a^N$. We get that $2r$ and $2\sqrt{a^N}r$ are overapproximations of the diameters of $Y_i, ~~i \leq N,$ and $Y_{N+1},$ respectively. Then, recalling the distortion \eqref{eq:abstraction_distortion_function}, for any trajectory $\xi$ of the system, $l > N,$
\[ 
    d(\xi,\Omega_A) \leq \frac{1}{l}(4Nr^2 + 4(l-N)a^{N}r^2) \underset{l\to\infty}{=} 4a^{N}r^2,
\]
which can be made arbitrarily small by suitable choice of $N$.
\end{example}
Further, as the following example shows, for Schur LTI systems, the bound \eqref{eq:slb_abstraction_distortion} converges to 0, for $l\to\infty$, which demonstrates the bound's tightness.
\begin{example}[Schur LTI systems] For an LTI system $x^+=Ax$, where the matrix $A$ is Schur-stable (its spectral radius is less than 1), we have $c_{\Bc^S_l} \leq v_n$ and
    \[
    \lim_{l\to\infty}J_{b_l}(x)\tran J_{b_l}(x) = \sum_{i=0}^{\infty}(A^i)\tran A^i = W,
    \] 
    where $W$ is the discrete observability Gramian \cite[Thm.~6.D1]{ctchen} for the pair $(A,I)$; i.e., $W$ is the solution to the Lyapunov equation $W-A\tran\, W A = I.$
Thus, both $h(\xi)$, $h_s(\xi)$ are finite, for $l\to\infty$, and the bound \eqref{eq:slb_abstraction_distortion} converges to 0.
\end{example}

Conversely, the example below shows that, even for marginally stable systems, the bound may not vanish with $l\to\infty$. 

\begin{example}[Marginally stable LTI system]
    For the simple system $x^+ = x$, we have that $\det(J_{b_l}(x)\tran J_{b_l}(x)) = \det(lI) = l^n$ and, by Prop.~\ref{prop:entropy}, $h(\xi) = h(x_0) + \frac{n\log l}{2}$ and 
    \[
    \begin{aligned}
    h_s(\xi) &= \frac{1}{1-s}\log\Big(\frac{1}{l^{n(s-1)/2}}\int_{\real^n}p_{0}^{s} \, \diff x\Big)\\
    &= -\frac{n(s-1)\log l}{2(1-s)} + \frac{1}{1-s}\log\int_{\real^n}p_{0}^{s} \, \diff x \\
    &=\frac{n\log l}{2} + h_{s}(\xi_0).
    \end{aligned}
    \] Replaced in the distortion bound \eqref{eq:slb_abstraction_distortion}, the $l$ in the denominator is canceled out, indicating a positive lower bound for any $l$. This independence on $l$ is expected, as abstracting $x^+ = x$ is the same as encoding the initial condition $\xi_0$.
\end{example}

\subsection{A relaxed lower bound and the curse of dimensionality}\label{ssec:curse}
We now derive relaxed lower bounds that do not depend on the dynamics, and thus do not require computing trajectory entropies, as the following corollary to Thm.~\ref{thm:slb_abstractions} shows.

\begin{corollary}[Relaxed lower bound for abstractions]\label{cor:relaxed}
    Consider a dynamical system $x^+=f(x)$ under the assumptions of Thm.~\ref{thm:slb_abstractions}, with $\xi_0\sim p_{\xi_0}:\Xc\to\real_+$. Assume that $f$ is either (piecewise) affine or Lipschitz continuous with constant $L$. The average distortion of any abstraction $A$ with partition size $|\Yc|\leq e^R$, where $R>0$, is lower bounded as follows:
    \begin{equation}\label{eq:slb_abstraction_distortion_contractive}
    \begin{aligned}
        D_{abs}(R)  \geq\, K(l)\bigg(&\frac{n}{2}\Big(\frac{e^{-R+h(\xi_0)-n/2}}{\Gamma(1+n/2)}\Big)^{2/n} \\&+\sup_{s\in (1,\infty]} e^{\frac{2}{n}(-\frac{s}{s-1}R+h_{s}(\xi_0))}\bigg),
    \end{aligned}
    \end{equation}
    where
    \begin{enumerate}
        \item $K(l) = \frac{1-L^2}{l{v_n}^{2/n}}$, if $f$ is Lipschitz with $L < 1$;
        \item $K(l) = \frac{1}{l^2{v_n}^{2/n}},$ if $f$ is Lipschitz with $L=1$;
        \item $K(l) = \frac{L^2-1}{lL^{2l}{v_n}^{2/n}}$ if $f$ is Lipschitz with $L>1$;
        \item $K(l) = \frac{1}{lM^{2l/n}{v_n}^{2/n}},$ if $f$ is piecewise affine with $M$ pieces.
    \end{enumerate}
    Moreover, picking $s=\infty$ yields:
    \begin{equation}\label{eq:rate lower bound}
    \begin{aligned}
        R_{abs}(D) \geq &\frac{n}{2}\log\!\Big(K(l)\frac{n}{2}\frac{e^{\frac{2}{n}h(\xi_0)-1}}{\Gamma(1+n/2)^{2/n}} + K(l)e^{\frac{2}{n}h_\infty(\xi_0)}\Big) \\&-\frac{n}{2}\log(D).
    \end{aligned}
    \end{equation}
\end{corollary}

While more conservative than Thm.~\ref{thm:slb_abstractions}, Cor.~\ref{cor:relaxed} confirms abstractions' curse of dimensionality: for any system and finite $l$, the partition size $e^R$ \emph{fundamentally increases exponentially with $n$} to achieve a prescribed distortion (by \eqref{eq:rate lower bound}).

\section{Numerical Examples}\label{sec:numerics}
\subsection{The doubling map}\label{ssec:chaotic map}

Consider the doubling map from Example \ref{ex:doubling map}. For any trajectory length $l$, its behavior $\Bc_l^S$ is composed by $2^{l-1}$ line segments in $\real^l$ described by $(x_0, 2x_0, 4x_0, ..., 2^{l-1}x_0) \bmod 1,$ uniformly distributed with  $p_\xi(\xi) = 1/\sqrt{1+4+...+4^{l-1}} = \sqrt{3/(4^l-1)},$ giving $h(\xi) = h_s(\xi) = \frac{1}{2}(\log(4^l-1)-\log{3})$ for all $s \in (1,\infty].$ Using Prop.~\ref{prop:c} for piecewise affine systems, we obtain $c_{\Bc_l^S} \leq \sqrt{(4^l-1)/3}v_n,$ enabling us to compute the lower bound in Theorem \ref{thm:slb_abstractions}. In light of Remark \ref{rem:c_piecewise}, we also determine the high-rate lower bound by picking $c_{\Bc_l^S} = v_1 = 2.$ The lower-bound curves can be seen in Fig.~\ref{fig:chaos boxes} (left) for $s = 2$ and $s = \infty$. It is apparent that the tightest bound is obtained with $s=\infty$ and $c_{\Bc_l^S} = v_1.$ In this case, the bound is consistently half of that of the optimal distortion $D_{abs}(R)$, which is remarkably close. As a comparison, the standard Shannon lower bound is~1.42 times smaller than the optimal quantizer distortion of a uniform random variable in $\real^1$, in the standard source-coding setting with Euclidean distortion.

Let us explain how we were able to compute the actual optimal achievable abstraction distortion. Following the reasoning in Section \ref{sec:main}, we first build an optimal cover for $\Bc_l^S$ (afterwards, we show that this optimal cover admits a distortion that is equal to that of a specific abstraction with the same rate, and thus its rate-distortion curve is optimal, among all abstractions). For a given $l$, consider $R = \log(k2^{l-1})$, where $k$ is an arbitrary natural number. Since all segments are equiprobable and congruent, and probability is uniform among them, the optimal cover of $\Bc_l^S$ is obtained by cutting each of the $2^{l-1}$ segments in $k$ equal pieces.\footnote{For uniform distributions, with convex and symmetric distortion functions as the one in \eqref{eq:abstraction_distortion_function}, optimal covers (or \emph{quantizers}) are uniform \cite{quantization}.} The cover is then the collection $M=\{C_i\}_{i=1}^{2^{l-1}k}$ of all these pieces---see Fig.~\ref{fig:chaos boxes} for an example with $l=3$ and $k=2$. Now let $C_\xi \in M$ denote the piece including the trajectory $\xi$. The expected error between $\xi$ and the Chebyshev center of $C_\xi$ is
$ \expect[\|\xi-x_c(C_\xi)\|^2] = \frac{1-4^{-l}}{9k^2},$
which is obtained by computing the squared length of each segment, $L^2 = 1/(2^{l-1}k)^2(1+2^2+...+(2^{l-1})^2) = (4^l-1)/(3k^24^{l-1}) = 4(1-4^{-l})/3k^2,$ followed by using the variance of the uniform distribution, giving $L^2/12$. Then, for the cover's set-based distortion, we have $d(\xi,C_\xi)=\frac{1}{l}\max_{\xi'\in C_\xi}\|\xi-\xi'\|^2 = \frac{1}{l}(\|\xi-x_c(C_\xi)\| + L/2)^2$, since $C_\xi$ is a straight-line segment. We, thus have $\expect[d(\xi,C_\xi)] = \frac{1}{l}\expect[(\|\xi-x_c(C_\xi)\| + L/2)^2]=7L^2/(12l)$, where we used that $X\coloneqq\|\xi-x_c(C_\xi)\|+L/2$ is a uniform random variable in $[L/2,L]$, and thus $E[X^2]= \mathrm{Var}(X)+\expect[X]^2 = L^2/48 + 9L^2/16 = 7L^2/12$. As $M$ is the optimal cover, we have $D_{cover}(R) = \expect[d(\xi,C_\xi)] = \frac{7}{36l}(4^l-1)e^{-2R},$
where we used $R = \log(k2^{l-1})$, and $D_{cover}(R)$ is the optimal distortion among all $nl$-dimensional covers of $\Bc_l^S$.
\begin{figure*}[t]
    \centering
    \includegraphics[height=3.5cm]{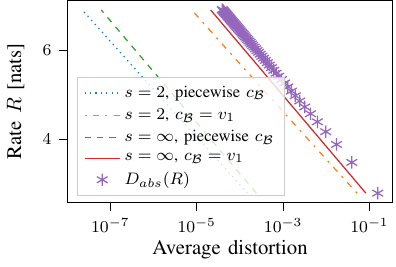}
    \hfill
    \includegraphics[height=3.6cm]{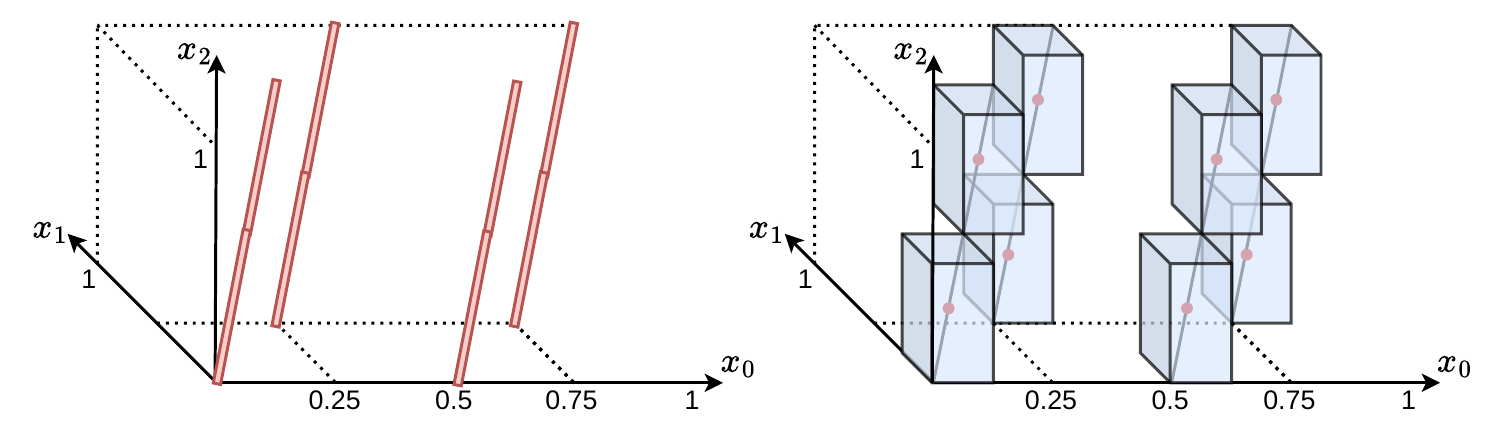}
    \caption{Section \ref{ssec:chaotic map}: Comparison between $D_{abs}(R)$ and the fundamental lower bound from Theorem \ref{thm:slb_abstractions}, for $l=5$ (left). Optimal cover of $\Bc_3^S$ with $k=2$ (center), and abstract reachable sets (right, blue boxes). On the right, the lines inside the boxes represent $\Bc_3^S$, with their Chebyshev centers marked in red. The maximal distance between any point in $\Bc_3^S$ and its corresponding blue box is obtained at one of the edges intersecting with the trajectory.}
    \label{fig:chaos boxes}
\end{figure*}

Finally, we show that, for $R = \log(k2^{l-1}), l = 1,2,...$,
$$D_{abs}(R) = D_{cover}(R) = \frac{7}{36l}(4^l-1)e^{-2R}.$$ Notice that, in general, $D_{abs}(R)\geq D_{cover}(R)$, as behaviors of abstractions are covers. However, the optimal cover built above admits an abstraction $A$ that gives the same distortion, and thus we have $D_{abs}(R)=D_{cover}(R)$. First, let $\Yc$ be the uniform grid of $[0,1]$ with segments of length $1/k2^{l-1}$. Each trajectory of the abstraction is a sequence of segments of lengths $1/k2^{l-1}, 1/k2^{l-2}, ..., 1/k$; thus, for any given trajectory $\xi$, the abstraction output $\Omega_A(\xi)$ is a box in $\real^l$ containing any related trajectory $\xi$ (see Fig.~\ref{fig:chaos boxes}). In fact, $\xi$ lies on one of the diagonals of the box $\Omega_A(\xi)$, and this diagonal is precisely $C_\xi$. We thus have $\arg\sup_{\xi'\in\Omega_A(\xi)}\|\xi-\xi'\|^2= \arg\sup_{\xi'\in C_\xi}\|\xi-\xi'\|^2$, and hence $d(\xi,\Omega_A(\xi)) = d(\xi,C_\xi)$. Consequently, the abstraction has the same distortion as the optimal cover.

\subsection{A 3D nonlinear system and uniform-grid abstractions}\label{ssec:3d nonlinear case}

Consider the nonlinear system $f: \real^3 \to \real^3$ where
\[ f(x) = \begin{bmatrix}
        0.9x_1 + 0.1\sin {x_2} \\
        2x_2^3 - x_2 \\
        0.9x_3 + 0.1x_1x_2,
\end{bmatrix}
\]
and $\Xc = [-1, 1]^3,$ which is forward invariant under $f$.  This system has multiple equilibria, hence the origin is not stable in $\Xc$. For each $N$ in $\{10, 20, 50, 100\},$ we build abstractions $A_N$ by using uniform partitioning of $\Xc$ with grids of size $N \times N \times N$ and determining the transition map using interval arithmetic.  Then, we compute the distortion lower bound from Theorem \ref{thm:slb_abstractions} using Prop.~\ref{prop:entropy} and Prop.~\ref{prop:c}, case 3.\footnote{The entropies were computed using Monte-Carlo integration with 10000 samples, while Jacobians and the Lipschitz constant were determined using automatic differentiation.} Furthermore, lower bounds were also computed by picking $c_{\Bc_l^S} = v_3$, in light of Remark \ref{rem:c_piecewise}. The resulting distortion lower bound curves can be seen in Fig.~\ref{fig:real abstraction rd}. In this case, as the abstraction we construct is not necessarily the optimal one, its expected distortion is generally 100x higher than the fundamental lower bound. Still, this demonstrates the validity of the lower bound, even in cases with complex, nonlinear dynamics; even more importantly, it indicates how conservative standard abstractions with uniform grids might be.

\begin{figure}[thb]
    \centering
    \includegraphics[width=0.8\linewidth]{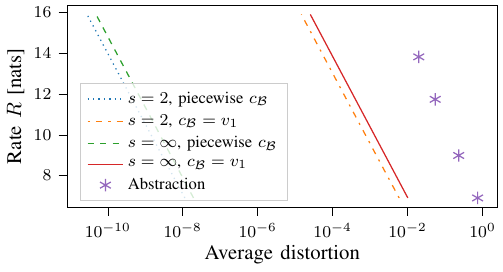}
    \caption{Section \ref{ssec:3d nonlinear case}: Comparison between average distortion ($l=5$) from the constructed abstractions and the fundamental lower bound from Theorem \ref{thm:slb_abstractions}.}
    \label{fig:real abstraction rd}
\end{figure}

\section{Conclusion and Future Research}
We have developed a statistical, quantitative theory on the size-accuracy tradeoff of finite abstractions of dynamical systems. Through this theory, we have uncovered fundamental limits on their scalability: given the dynamics, we have obtained a fundamental bound on the achievable abstraction accuracy, for a given abstraction size; and, conversely, a lower bound on the abstraction size, for a prescribed accuracy. To that end, we have established connections with rate-distortion theory. From an information-theoretic perspective, we have developed rate-distortion theory for the class of encoder-decoder pairs that abstractions constitute: set-based, with set-based distortion. This novel theory quantifies scalability limits of abstractions, and provides insights on how the complexity of the dynamics to be abstracted dictates these limits. 

Importantly, the developed theory may be employed to construct minimal abstractions, harnessing their full scalability potential. This work made clear that, to construct minimal abstractions, one has to encode system trajectories, through covers in a high-dimensional ambient space. This is demonstrated in Section \ref{ssec:chaotic map}, where we construct a minimal abstraction of the doubling-map dynamics. Future research will thus focus on the general problem of constructing minimal abstractions. 

Further, extending our theory to control-oriented abstractions, e.g. alternating simulations \cite[Def. 9.6]{tabuada2009verification} is a topic of future research. A main challenge is incorporating inputs in the information-theoretic framework. Although one may cast inputs the same way as states, and proceed with extending our theory, this might yield conservative results. A more sensible approach might be capturing the causal relationship between control and trajectories, through causal rate-distortion theoretic concepts, e.g. non-anticipative rate-distortion functions \cite{charalambous2013nonanticipative}.

\section{Technical Results and Proofs}\label{sec:proofs}

\begin{proof}[\textbf{Proof of Prop. \ref{prop:abstraction vs encoder distortion}}]
    For any given $\xi\in \Bc_{l}^S$, we will prove that
    \begin{equation*}
        d(\xi,\Omega_A) \geq \frac{1}{l}\|\xi-\xi_{q_A}\|^2 + \frac{1}{l}r_c^2(\Omega_A),
    \end{equation*}
    where note that $\xi\in \Omega_A$ (by behavioral inclusion, c.f. Thm. \ref{thm:inclusion}) and $\xi_{q_A}=x_c(\Omega_A)$. Then, the proof is complete by applying the expectation operator to the above inequality.

    Notice that, since \(\Omega_A\subseteq\mathcal X^l\) is bounded, \(\overline{\Omega_A}\) is compact, and taking the closure does not change suprema/infima taken over $\Omega_A$ or the Chebyshev radius, that is $\sup_{y\in\Omega_A}\|y-x'\|=\max_{y\in\overline{\Omega_A}}\|y-x'\|$, for any $x'$. Hence we equivalently work here with \(\overline{\Omega_A}\), for which maxima are attained. Define
    $w(x') = \max_{y\in \overline{\Omega_A}} \|y-x'\|^2
    $. The function $w$ is convex, being the pointwise maximum of the convex quadratic maps
    $x'\mapsto\|y-x'\|^2$. We have
    $x_c(\Omega_A) = \arg\min_{x'\in\mathbb R^{nl}} w(x')
    $ and
    $r_c^2(\Omega_A) = w(x_c(\Omega_A)) = \max_{y\in \overline{\Omega_A}}\|y-x_c(\Omega_A)\|^2$.
    
    Define the set of maximizers
    \[
    M := \{y\in \overline{\Omega_A} : \|y - x_c(\Omega_A)\| = r_c(\Omega_A)\}.
    \]
    The subdifferential of $w$ at $x_c(\Omega_A)$ is $\partial w(x_c(\Omega_A))
    = 2\operatorname{conv}\{\,x_c(\Omega_A)-y : y\in M\,\}$, where $\operatorname{conv}$ denotes the convex hull operator. Since $x_c(\Omega_A)$ minimizes $w$, the optimality condition
    $0\in\partial w(x_c(\Omega_A))$ gives
    $
    0 \in \operatorname{conv}\{\,x_c(\Omega_A)-y : y\in M\,\}
    $. Hence there exist finitely many points $y_1,\dots,y_m\in M$ and
    coefficients $\lambda_i\ge0$, $\sum_i\lambda_i=1$, such that
    \begin{equation}\label{eq:convexcomb}
    \sum_{i=1}^m \lambda_i (y_i - x_c(\Omega_A)) = 0.
    \end{equation}
    
    Now, fix $\xi \in \Bc_l^S$ and let
    $\xi_* \in \arg\max_{y\in \overline{\Omega_A}} \|y - \xi\|^2$. By definition of $\xi_*$, for every $y_i\in M$ we have $\|y_i - \xi\|^2 \le \|\xi_* - \xi\|^2$.
    Taking the convex combination with the $\lambda_i$ gives
    \[
    \sum_{i=1}^m \lambda_i\big(\|y_i-\xi\|^2 - \|\xi_* - \xi\|^2\big) \le 0.
    \]
    Since { $\|y_i - \xi\|^2 = \|y_i - x_c(\Omega_A)\|^2
       + \|x_c(\Omega_A) - \xi\|^2
       + \allowbreak 2 (y_i - x_c(\Omega_A))\tran(x_c(\Omega_A) - \xi)$}
    and $\|y_i - x_c(\Omega_A)\|^2 = r_c(\Omega_A)^2$, for the above inequality we have
    \[
    r_c(\Omega_A)^2 + \|x_c(\Omega_A)-\xi\|^2
      - \|\xi_* - \xi\|^2 \le 0,
    \]
    where the cross term has vanished, using \eqref{eq:convexcomb}.
    Finally, using \eqref{eq:abstraction_distortion_function},
    \[
    r_c(\Omega_A)^2 + \|\xi_{q_A}-\xi\|^2 \le \|\xi_* - \xi\|^2=l\cdot d(\xi,\Omega_A).
    \]
\end{proof}

Towards proving Thm. \ref{thm:slb_abstractions}, we use the following lemma.
\begin{lemma}\label{lem:sphere_packing}
Let $M\subset \mathbb{R}^n$ be a finite union of bounded, $m$-dimensional $C^1$-manifolds. Let $X$ be a random variable in $M$ with probability measure $\mu_X\ll\Hc^m_{M}$ and density $p= \frac{\diff \mu_x}{\diff \Hc^m_{M}}$. Consider a finite measurable cover $\Rc \coloneqq \{R_i\}_{i=1}^N$ of $M$, and a measurable function $w:M\to\Rc$, which is such that $x\in w(x)$ for all $x\in M$. Then:

\begin{equation}\label{eq:sphere packing inequality}
\begin{aligned}
\mathbb{E}_{X } \Big[ r_c(w(X))^2 \Big]
\ge c_M^{-2/m} \sup_{s\in (1,\infty]} e^{\frac{2}{m}h_{s}(X)}\, N^{-\frac{2}{m(1-1/s)}}&,
\end{aligned}
\end{equation}
where $c_M$ is defined by \eqref{eq:c}, and $r_c(R_i)$ denotes the Chebyshev radius of $R_i$.
\end{lemma}

\begin{proof}
Define $S_i:=\{x\in\mathcal M:w(x)=R_i\}=w^{-1}(\{R_i\})$. Since \(w\) is measurable and single-valued, the sets \(S_i\) are measurable and pairwise disjoint. Moreover, since \(w(x)\in\mathcal R\) for all \(x\in\mathcal M\), they cover \(\mathcal M\). Thus, \(\{S_i\}_{i=1}^N\) forms a measurable partition of \(\mathcal M\). Furthermore, since \(x\in w(x)\), we have \(S_i\subseteq R_i\) for every \(i\). Let
\[
p_i:=\mu_X(S_i)=\int_{S_i}p\,d\mathcal H^m,
\]
so that \(\sum_{i=1}^N p_i=1\), and let \(r_i:=r_c(S_i)\) whenever \(S_i\neq\emptyset\) (set \(r_i:=0\) otherwise). Since \(S_i\subseteq R_i\), $r_i\le r_c(R_i)$. Thus,
\[
\begin{aligned}
\mathbb E_X[r_c(w(X))^2]
=\sum_{i=1}^N p_i\,r_c(R_i)^2
\ge \sum_{i=1}^N p_i r_i^2.
\end{aligned}
\]
Hence, it suffices to lower bound \(\sum_{i=1}^N p_i r_i^2\).

\vspace{0.3em}

For a given $i$, by definition, $S_i \subset \overline{B}(c_i, r_i) \cap M$ for some Chebyshev center $c_i$. For any $s>1$, we have:
\begin{align*}
    p_i &\le \int_{\overline{B}(c_i,r_i)\cap M} p \, d\mathcal{H}^m\\
    &=\int_{M}p\,\mathbf{1}_{\overline{B}(c_i,r_i)\cap M}\,d\Hc^m\\ 
    &\le \Big(\int_Mp^s d\Hc^m\Big)^{1/s}\Big(\int_M(\mathbf{1}_{\overline{B}(c_i,r_i)\cap M})^{\frac{s}{s-1}}\,d\Hc^m\Big)^{1-1/s}\\
    &\le \|p\|_{L^s(M)} \, (\mathcal{H}^m(\overline{B}(c_i,r_i)\cap M))^{1-1/s}\\
    &\le \|p\|_{L^s(M)}\, (c_{M}\,r_i^m)^{1-1/s},
\end{align*}
where $\|p\|_{L^s(M)}\coloneqq \Big(\int_Mp^s d\Hc^m\Big)^{1/s},$ in the third step we used Hölder's inequality, and in the final step we used the inequality \eqref{eq:c}.\footnote{Although \eqref{eq:c} is defined for open balls, it also holds for closed balls, 
since \(\overline B(c_i,r_i)\subset B(c_i,r_i+\varepsilon)\) for every \(\varepsilon>0\), and one may let \(\varepsilon\downarrow0\).} Defining $
K_s \coloneqq \|p\|_{L^s(M)} \, c_{M}^{1-1/s}$, from the inequality above we have:
\[
r_i \ge \left( \frac{p_i}{K_s} \right)^{\frac{1}{m(1-1/s)}} \implies r_i^2 \ge \left( \frac{p_i}{K_s} \right)^\alpha,
\]
where $\alpha := 2/(m(1-1/s)) = 2s/(m(s-1))$. Multiplying by $p_i$ gives
\begin{equation}\label{eq:pr2_first_lower_bound}
p_i r_i^2 \ge K_s^{-\alpha} \, p_i^{1+\alpha} \implies \sum_{i=1}^Np_i r_i^2 \ge K_s^{-\alpha}\sum_{i=1}^Np_i^{1+\alpha}.
\end{equation}
Our job now is to find a lower bound to $\sum_{i=1}^Np_i^{1+\alpha}$ over discrete probabilities $p_i$. First, notice that $\alpha >0$ since $s > 1$. Therefore, the map $t \mapsto t^{1+\alpha}$ is convex in $t\in[0,+\infty)$. Thus, by Jensen's inequality,
\[
\sum_{i=1}^N p_i^{1+\alpha} \ge N\bigg(\frac{1}{N}\sum_{i=1}^N p_i\bigg)^{1+\alpha} = N \left( \frac{1}{N} \right)^{1+\alpha} = N^{-\alpha}.
\]
Substituting in \eqref{eq:pr2_first_lower_bound} gives
\[
\sum_{i=1}^N p_i r_i^2 \ge (K_s N)^{-\alpha} = \Big(\int_Mp^s d\Hc^m\Big)^{-\alpha/s}N^{-\alpha} \, c_M^{-2/m}.
\]

Now, by definition of the Rényi entropy, we get
\[ h_s(x) = \frac{1}{1-s}\log\expect\Big[p(x)^{s-1}\Big] = \frac{1}{1-s}\log\!\!\int\!\!\Big(\frac{\diff \mu_x}{\diff \Hc_\Xc^m}\Big)^{s-1}\diff\mu_x,\]
which, by the properties of the Radon--Nikodym derivative, gives
\[
\begin{aligned}
    h_s(x) &= \frac{1}{1-s}\log\int\Big(\frac{\diff \mu_x}{\diff \Hc_\Xc^m}\Big)^{s}\diff\Hc_\Xc^m \\
    &= \frac{1}{1-s}\frac{s}{\alpha}\frac{\alpha}{s}\log\int p^{s}\diff\Hc_\Xc^m \\
    &= -\frac{s}{\alpha(1-s)}\log\Big(\int p^{s}\diff\Hc_\Xc^m\Big)^{-\alpha/s}.
\end{aligned}
\]
And, using the fact that $\alpha(1-s) = -2s/m$, we get
\[
\begin{aligned}
&h_s(x) = \frac{m}{2}\log\Big(\int p^{s}\diff\Hc_\Xc^m\Big)^{-\alpha/s} \\
&\iff \Big(\int p^{s}\diff\Hc_\Xc^m\Big)^{-\alpha/s} = \mathrm{e}^{\frac{2}{m}h_s(x)}.
\end{aligned}
\]

Therefore, for any $s > 1$, we have
\[
\mathbb{E}_{X } \Big[ r_c(w(X))^2 \Big]
\ge c_M^{-2/m} \, e^{\frac{2}{m}h_{s}(X)}\, N^{-\frac{2}{m(1-1/s)}}.\]
Since this holds for every \(s\in(1,\infty)\), and \(h_s(X)\to h_\infty(X)\) as \(s\to\infty\), taking the supremum over \(s\in(1,\infty]\) yields \eqref{eq:sphere packing inequality}.
\end{proof}

\begin{proof}[\textbf{Proof of Thm. \ref{thm:slb_abstractions}}]
    We make use of Prop. \ref{prop:abstraction vs encoder distortion}. Take \eqref{eq:prop abstraction vs encoder} and minimize both sides over all possible partitions $\Yc$ with size $|\Yc|\leq e^R$ and associated abstractions $A$. We have
    \begin{equation*}
        D_{abs}(R)\geq \inf_{A, |\Yc|\leq e^R}\frac{1}{l}\expect_{\xi}[\|\xi-\xi_{q_A}\|^2] + \frac{1}{l}\expect_{\xi}[r_c^2(\Omega_A)],
    \end{equation*}
    where recall that $\expect_{\xi_0}[\cdot] = \expect_{\xi}[\cdot]$, and that for a given abstraction $A$ with corresponding encoder-decoder pair $s_A,g_A$, we have $\xi_{q_A} = g_{q_A}(s_{q_A}(\xi))$ with  $s_{q_A}(\xi) = g_A(s_A(\xi))$ and $g_{q_A}(z) = x_c(z)$, where $x_c(z)$ is the Chebyshev center of the set $z$; and $r_c(\Omega_A)$ is the Chebyshev radius of $\Omega_A$. Notice that $(s_{q_A},g_{q_A})$ is an encoder-decoder pair with rate $R$ (the same as the abstraction) and $\xi_{q_A}$ is its output. Hence, the first term in the left-hand side of the above inequality can be lower bounded by employing Thm. \ref{thm:slb}, to obtain: 
    \begin{equation*}
        D_{abs}(R)\geq \frac{n}{2l}\Big(\frac{e^{-R+h(\xi)-n/2}}{c_{\Bc_l^S}\Gamma(1+n/2)}\Big)^{2/n} + \frac{1}{l}\inf_{A, |\Yc|\leq e^R} \expect_{\xi}[r_c^2(\Omega_A)].
    \end{equation*}
    To bound the second term, we employ Lemma \ref{lem:sphere_packing}. Notice that the abstraction's outputs $\Omega_A\in\Rc_l^A$ are elements of the $nl$-dimensional cover $\Rc_l^A$ of $\Bc_l^S$ (which is $n$-dimensional) with cardinality $|\Rc_l^A|=|\Yc|\leq e^R$. Thus, the term $\inf_{A, |\Yc|\leq e^R} \expect[r_c^2(\Omega_A)]$ can be lower bounded as in \eqref{eq:sphere packing inequality}, where we replace $M$ by $\Bc_l^S$, $\Rc_l^A$ by $\Rc$, $m$ by $n$, $N$ by $|\Yc|$, $\mu_X$ by $\mu_{\xi}$, $w(X)$ by $g_A(s_A(\xi))=:\Omega_A$, and take $|\Yc|\leq e^R$, to obtain:
    \begin{equation*}
        \frac{1}{l}\inf_{A, |\Yc|\leq e^R} \expect_{\xi}[r_c^2(\Omega_A)]\geq \frac{1}{l} c_{\Bc_l^S}^{-2/n} \sup_{s\in (1,\infty]} e^{\frac{2}{n}(-\frac{s}{s-1}R+h_{s}(\xi))} 
    \end{equation*}
\end{proof}

\begin{proof}[\textbf{Proof of Prop. \ref{prop:entropy}}]
    Fix any measurable subset $\Ac \subseteq \Xc$. Because $\mu_{\xi_0}(\Ac) = \mu_\xi(b_l(\Ac)),$ the definitions of $p_\xi$ and $p_{\xi_0}$ imply that
    \[ 
    \int_{b_l(A)}p_\xi(y)\diff \Hc^n_{\Bc^S_l}(y)
    = \int_{\Ac}p_{\xi_0}(x)\diff \Lc^n(x). 
    \]
    But also, since the first $n$ components of $b_l(x)$ equal $x$ itself, we have that $b_l$ is injective, allowing the use of the area formula \cite[Thm.~3.2.5]{federer1969geometric} to get 
    \[
    \int_{b_l(\Ac)}\hspace{-2mm}p_\xi(y)\diff \Hc^n_{\Bc^S_l}(y) = \int_\Ac p_{\xi}(b_l(x))\sqrt{\det(J_{b_l}(x)\tran J_{b_l}(x))}\diff \Lc^n
    \]
    implying that, for almost all $x \in \Xc,$

    \begin{equation}\label{eq:behavior density}
    p_\xi(b_l(x)) = \frac{p_{\xi_0}(x)}{\sqrt{\det(J_{b_l}(x)\tran J_{b_l}(x))}}.
    \end{equation}

    The expression above is well-defined everywhere as $J_{b_l}\tran J_{b_l} = \eye + J_f\tran J_f + \cdots \succeq \eye,$ implying that $\det (J_{b_l}\tran J_{b_l}) \geq 1$.

    Then, \eqref{eq:generalized entropy} becomes
    \[
    \begin{aligned}
        h(\xi) =& 
        -\int_{\real^n}p_{\xi_0}(x)\log(p_{\xi_0}(x))\diff \Lc^n \\ &+\frac{1}{2}\int_{\real^n}\log\det(J_{b_l}(x)\tran J_{b_l}(x))p_{\xi_0}(x)\diff \Lc^n.
    \end{aligned}
    \]
    Likewise, the area formula gives
    \begin{align*}
        h_{s}(\xi) =& \frac{1}{1-s} \log\int_{\Bc_l^S} p_{\xi}^{s} \, \diff \mathcal{H}^n \\
        =& \frac{1}{1-s} \log\int_{\real^n} p_{\xi}(b_l(x))^{s} \, \det(J_{b_l}(x)\tran J_{b_l}(x))^{\frac{1}{2}}\diff \Lc^n,
    \end{align*}
    and, applying \eqref{eq:behavior density} gives \eqref{eq:renyi trajectory}. In the case of $s = \infty,$ we have
    \[
    \begin{aligned}
    h_{s}(\xi) &= \frac{s}{1-s} \log \Big(\int_{\real^n} \frac{p_{\xi_0}(x)^{s}}{\det(J_{b_l}(x)\tran J_{b_l}(x))^{\frac{s-1}{2}}} \, \diff x \Big)^{1/s} \\
    &\underset{s\to\infty}{=} -\log \ess \sup \frac{p_{\xi_0}(x)}{\sqrt{\det(J_{b_l}(x)\tran J_{b_l}(x))}}\\
    &= \ess \,\inf \log \frac{\sqrt{\det(J_{b_l}(x)\tran J_{b_l}(x))}}{p_{\xi_0}(x)},
    \end{aligned}
    \]
    which gives \eqref{eq:renyi max trajectory} by the properties of $\log$.
    
    Finally, since $\det (J_{b_l}\tran J_{b_l}) \geq 1$, we establish that $h(\xi) \geq h(\xi_0)$ and $h_s(\xi) \geq h_s(\xi_0)$ for any $s > 1.$
\end{proof}

\begin{lemma}\label{lem:lipschitz ball}
Let $X \subset \mathbb{R}^n,$ and $w : X \to \mathbb{R}^N$, $N \geq n,$ be an $L$-Lipschitz function satisfying $w(x) = [x\tran \ \ g(x)]\tran$ for some function $g$.
Then for every $y \in \mathbb{R}^N$ and $\delta > 0$,
\[
\Hc^n(w(X) \cap B(y,\delta)) \leq L^n v_n \delta^n.
\]
\end{lemma}

\begin{proof}
Fix $y \in \mathbb{R}^N$ and $\delta > 0$ and define $Z \coloneqq w(X) \cap B(y,\delta)$ and its pre-image $E \coloneqq w^{-1}(Z) \subset \real^n$. Let $P: \real^N \to \real^n$ denote the Euclidean projection onto the first $n$ coordinates. Because $w(x) = [x\tran \ \ g(x)]\tran$, we have that $P(w(x)) = x$ for all $x \in X.$ Thus, for any $x \in E,$ since $w(x) \in B(y,\delta),$ we get \(\|x - P(y)\| = \|P(w(x))-P(y)\| \leq \|w(x)-y\|<\delta,\) where the first inequality is because orthogonal projections are nonexpansive. Hence, $E \subseteq B(P(y),\delta)$, and consequently $\Hc^n(E) \leq v_n \delta^n.$ Since $w$ is $L$-Lipschitz, by fundamental properties of the Hausdorff measure \cite[Sec.~2.10.11]{federer1969geometric},
\[
\Hc^n(Z) = \Hc^n(w(E)) \leq L^n \Hc^n(E) \leq L^n v_n \delta^n.
\]
\end{proof}

\begin{proof}[\textbf{Proof of Prop. \ref{prop:c}}]
    We again use the function $b_l:\Xc\to\Bc_l^S$, defined by \eqref{eq:b}. Since by assumption $\Xc$ is full dimensional in $\real^n$, the tightest value for $c_\Xc$ is $c_{\real^n} = v_n.$ Now we look at each case.

    Case (1) follows trivially by the observation that $\Bc_l^S$ is an $n$-dimensional affine subset of $\real^{nl}$, and that the intersection of an $nl$-ball of radius $r$ and a plane of dimension $n$ is a ball of dimension $n$ and radius $\leq r.$ Hence, $\Hc_{\Bc_l^S}(B(z,\delta)\cap \Bc_l^S) \leq v_n\delta^n$, for all $z\in\real^{nl}$.

    Case (2): If $f$ is piecewise affine, so is $\Bc_l^S$, which has at most $M^l$ disjoint pieces. Denote by $Z_i$ each such piece of $\Bc_l^S$, which is a bounded, connected $n$-dimensional subset of some affine subspace of $\real^{nl}.$ Thus, $\Bc_l^S = \bigcup_{i=1}^{N}Z_i$, with $N\leq M^l$. Then, for all $z \in \real^{nl}$ and $\delta > 0,$ 
    \[
    \Hc^n\Big(\bigcup_iZ_i\cap B(z,\delta)\Big) 
    = \sum_{i=1}^{N}\Hc^n(Z_i\cap B(z,\delta))
    \leq M^lv_n\delta^n,
    \]
    where in the last inequality we used case (1) and $N\leq M^l$.

    Case (3): It is easy to see that $b_l$ satisfies the conditions of Lemma \ref{lem:lipschitz ball} with Lipschitz constant $\big(\sum_{i=0}^{l-1} L^{2i}\big)^{1/2}$. Hence the result comes from applying Lemma \ref{lem:lipschitz ball}.
\end{proof}

\begin{proof}[\textbf{Proof of Cor. \ref{cor:relaxed}}]
    We start by applying the property that $h(\xi) \geq h(\xi_0)$ and $h_s(\xi) \geq h_s(\xi_0)$ from Prop.~\ref{prop:entropy} in \eqref{eq:slb_abstraction_distortion}, Thm.~\ref{thm:slb_abstractions}. The fact that the RHS of \eqref{eq:slb_abstraction_distortion} is monotonically increasing with $h(\xi)$ and $h_s(\xi)$ gives
    \begin{equation}\label{eq:slb_xi0}
    \begin{aligned}
        D_{abs}(R)  \geq\, &\frac{n}{2l}\Big(\frac{e^{-R+h(\xi_0)-n/2}}{c_{\Bc_l^S}\Gamma(1+n/2)}\Big)^{2/n} \\&+\frac{1}{l} c_{\Bc_l^S}^{-2/n} \sup_{s\in (1,\infty]} e^{\frac{2}{n}(-\frac{s}{s-1}R+h_{s}(\xi_0))}.
    \end{aligned}
    \end{equation}
    Define $\tilde{K}(l) \coloneqq \frac{1}{l}\big(\frac{1}{c_{\Bc_l^S}}\big)^{2/n}.$ Applying Prop. \ref{prop:c} case (2) to $\tilde{K}(l)$ gives, for the piecewise affine case with $M$ pieces,
    \[ \tilde{K}(l) \geq K(l) = (lM^{2l/n}{v_n}^{2/n})^{-1}, \]
    which is case (4) of the corollary. For the remaining cases, we apply Prop. \ref{prop:c}, case (3), and use the fact that $\sum_{i=0}^{k-1}a^i = k$ for $a=1$, and otherwise $\sum_{i=0}^{k-1}a^i = (a^k - 1)/(a-1) \leq \max(1, a^k)/|a-1|$. Thus, we bound $\tilde{K}(l)$ accordingly via
    \begin{enumerate}
        \item $\tilde{K}(l) \geq \frac{1-L^2}{l{v_n}^{2/n}}$, if $f$ is Lipschitz with $L < 1$;
        \item $\tilde{K}(l) \geq \frac{1}{l^2{v_n}^{2/n}},$ if $f$ is Lipschitz with $L=1$;
        \item $\tilde{K}(l) \geq \frac{L^2-1}{lL^{2l}{v_n}^{2/n}}$ if $f$ is Lipschitz with $L>1$.
    \end{enumerate}
    The final expression \eqref{eq:slb_abstraction_distortion_contractive} is then obtained by applying the respective bounds on $\tilde{K}(l)$ to \eqref{eq:slb_xi0}, which is then inverted for the $s=\infty$ case to arrive at \eqref{eq:rate lower bound}.
\end{proof}

\bibliography{refs.bib}
\bibliographystyle{IEEEtran}

\end{document}